\documentclass{article}

\usepackage[backend=biber]{biblatex}
\usepackage{auxhook}
\usepackage{amsthm}
\usepackage{mathtools}
\usepackage{amssymb}

\theoremstyle{definition}
\newtheorem{definition}{Definition}

\newtheorem{theorem}{Theorem}
\newtheorem{corollary}{Corollary}

\begin{document}

\title{Triple Graph Grammars for Multi-version Models}

\author{Matthias Barkowsky\\matthias.barkowsky@hpi.de \and Holger Giese\\holger.giese@hpi.de}

\maketitle

\begin{abstract}
Like conventional software projects, projects in model-driven software engineering require adequate management of multiple versions of development artifacts, importantly allowing living with temporary inconsistencies. In the case of model-driven software engineering, employed versioning approaches also have to handle situations where different artifacts, that is, different models, are linked via automatic model transformations.

In this report, we propose a technique for jointly handling the transformation of multiple versions of a source model into corresponding versions of a target model, which enables the use of a more compact representation that may afford improved execution time of both the transformation and further analysis operations. Our approach is based on the well-known formalism of triple graph grammars and a previously introduced encoding of model version histories called multi-version models. In addition to showing the correctness of our approach with respect to the standard semantics of triple graph grammars, we conduct an empirical evaluation that demonstrates the potential benefit regarding execution time performance.
\end{abstract}


\section{Introduction} \label{sec:introduction}

In model-driven software development, models are treated as primary development artifacts. Complex projects can involve multiple models, which describe the system under development at different levels of abstraction or with respect to different system aspects and can be edited independently by a team of developers. In this case, consistency of the holistic system description is ensured by model transformations that automatically derive concrete models from more abstract ones or propagate changes to a model describing one aspect of the system to other models concerned with different but overlapping aspects \cite{seibel2010dynamic}.

Similarly to program code in conventional software development, the evolution of models via changes made by different developers requires management of the resulting versions of the software description. In particular, version management has to support parallel development activities of multiple developers working on the same development artifact, where living with inconsistencies of a single artifact may temporarily be necessary to avoid loss of information \cite{Finkelstein+1994}. In \cite{Barkowsky2022}, we have introduced multi-version models as a means of managing multiple versions of the same model that also enables monitoring the consistency of the individual model versions and potential merge results of versions developed in parallel.

However, with model transformations effectively linking multiple models via consistency relationships, considering only the evolution of a single model without its context is insufficient for larger model-driven software development projects. Thus, a mechanism for establishing consistency of different versions of such linked models that simultaneously allows parallel development of multiple versions is required.

Such a mechanism would allow working with more compact representations that also enable further analysis operations as described in \cite{Barkowsky2022}. In addition, an integrated handling of multiple model versions may afford an improved execution time performance of the transformation.

In this report, we propose a first step in the direction of model transformations working on multi-version models by adapting the well-known formalism of triple graph grammars, which enables the implementation of single-version model transformations, to the multi-version case.

The remainder of the report is structured as follows: In Section \ref{sec:preliminaries}, we briefly reiterate the basic concepts of graphs, graph transformations, triple graph grammars, and multi-version models, as used in this report. Subsequently, we present our approach for deriving transformation rules that work on multi-version models from single-version model transformation specifications in the form of triple graph grammars in Section \ref{sec:mv_rules_derivation}. In Section \ref{sec:mv_transformation_execution}, we describe how the derived rules can be used to realize the joint transformation of all individual model versions encoded in a multi-version model and prove the correctness of our technique with respect to the semantics of triple graph grammars. Section \ref{sec:evaluation} reports on the results of an initial evaluation of the presented solution's performance regarding execution time, which is based on an application scenario in the software development domain. Related work is discussed in Section \ref{sec:related_work}, before Section \ref{sec:conclusion} concludes the report.


\section{Preliminaries} \label{sec:preliminaries}

In this section, we give a brief overview of required preliminaries regarding graphs and graph transformations, triple graph grammars and multi-version models.

\subsection{Graphs and Graph Transformations}


We briefly reiterate the concepts of graphs, graph morphisms and graph transformations and their typed analogs as defined in \cite{Ehrig+2006} and required in the remainder of the report.

A graph $G = (V^G, E^G, s^G, t^G)$ consists of a set of nodes $V^G$, a set of edges $E^G$ and two functions $s^G: E^G \rightarrow V^G$ and $t^G: E^G \rightarrow V^G$ assigning each edge its source and target, respectively. A graph morphism $m: G \rightarrow H$ is given by a pair of functions $m^V: V^G \rightarrow V^H$ and $m^E: E^G \rightarrow E^H$ that map elements from $G$ to elements from $H$ such that $s^H \circ m^E = m^V \circ s^G$ and $t^H \circ m^E = m^V \circ t^G$. We also call $m^V$ the \emph{vertex morphism} and $m^E$ the \emph{edge morphism}.

A graph $G$ can be typed over a type graph $TG$ via a typing morphism $\mathit{type}: G \rightarrow TG$, forming the typed graph $G^T = (G, \mathit{type}^G)$. In this report, we consider a model to be a typed graph, with the type graph defining a modeling language by acting as a metamodel.

A typed graph morphism between two typed graphs $G^T = (G, \mathit{type}^G)$ and $H^T = (H, \mathit{type}^H)$ with the same type graph then denotes a graph morphism $m^T: G \rightarrow H$ such that $\mathit{type}^G = \mathit{type}^H \circ m^T$. A (typed) graph morphism $m$ is a monomorphism iff its functions $m^V$ and $m^E$ are injective.

Figure \ref{fig:example_graph_type_graph} shows an example typed graph $G$ from the software development domain along with the corresponding type graph $TG$. The typing morphism is encoded by the node's labels. $G$ represents an abstract syntax graph of a program written in an object-oriented programming language, where nodes may represent class declarations ($ClassDecl$), field declarations ($FieldDecl$) or type accesses ($TypeAccess$). Class declarations may contain field declarations via edges of type $declaration$, whereas field declarations can reference a class declaration as the field type via a $TypeAccess$ node and edges of type $access$ and $type$. The example graph contains two class declarations, one of which contains a field declaration, the field type of which is given by the other class declaration.

\begin{figure}
\centering
\includegraphics[width=\textwidth]{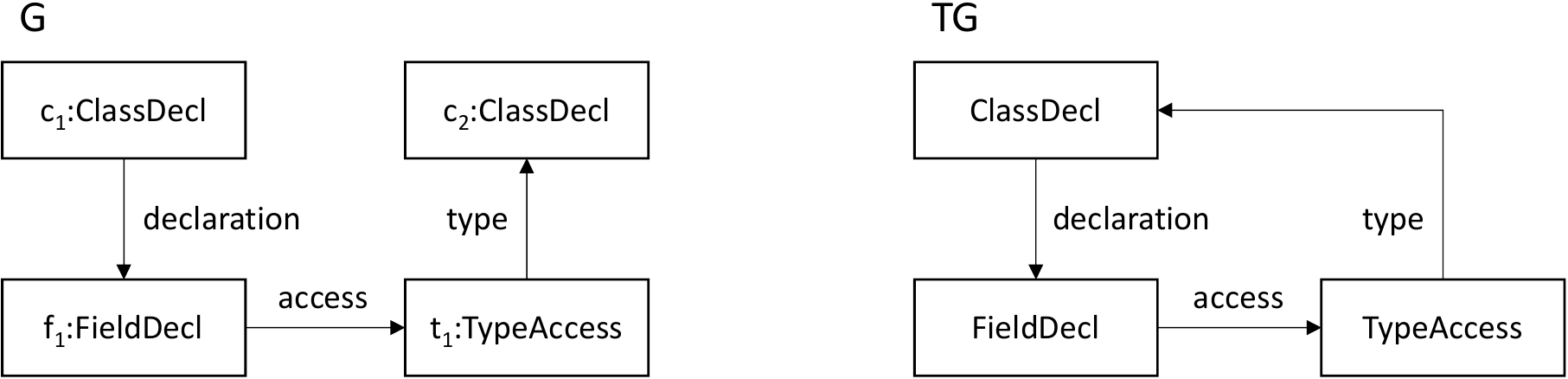}
\caption{example graph and type graph} \label{fig:example_graph_type_graph}
\end{figure}

A (typed) graph transformation rule $r$ is characterized by a span of (typed) graph monomorphisms $L \xleftarrow{l} K \xrightarrow{r} R$ and can be applied to a graph $G$ via a monomorphism $m : L \rightarrow G$ called match that satisfies the so-called dangling condition \cite{Ehrig+2006}. The result graph $H$ of the rule application is then formally defined by a double pushout over an intermediate graph \cite{Ehrig+2006}. Intuitively, the application of $r$ deletes the elements in $m(L)$ that do not have a corresponding element in $R$ and creates new elements for elements in $R$ that do not have a corresponding element in $L$. The graph $L$ is also called the rule's \emph{left-hand side}, $K$ is called the rule's \emph{glueing graph}, and $R$ is called the \emph{right-hand side}.

$r$ is called a graph production if it does not delete any elements, that is, $l$ is surjective. In this case, since $L$ and $K$ are isomorphic with $l$ an isomorphism and we only distinguish graphs up to isomorphism, we also use the simplified representation $L \xrightarrow{r} R$.

Figure \ref{fig:example_gt_rule} shows an example graph production in shorthand notation, where preserved elements are colored black, whereas created elements are colored green and marked by an additional ``++'' label. For two existing classes, the production creates a field declaration in one of them that references the other class declaration as the field type.

\begin{figure}
\centering
\includegraphics[width=0.4\textwidth]{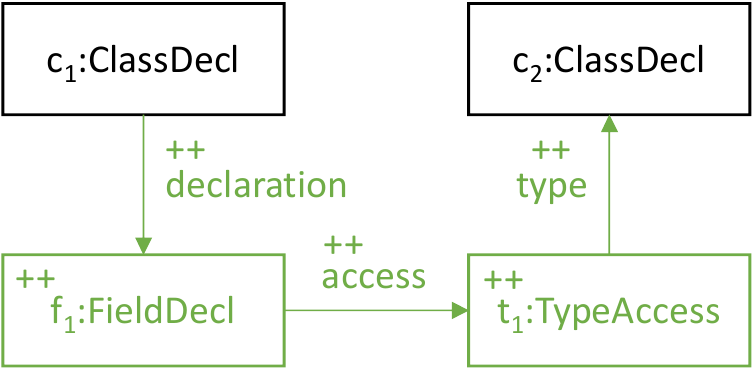}
\caption{example graph transformation rule in shorthand notation} \label{fig:example_gt_rule}
\end{figure}

We denote a sequence of applications of rules from a set of rules $R$ to a graph $G$ with resulting graph $G'$ by $G \rightarrow^{R} G'$. We say that such a rule application sequence is maximal if it cannot be extended by any application of a rule from $R$.

\begin{definition} \textit{Maximal Rule Application Sequence}
A sequence of rule applications $G \rightarrow^{R} G'$ with a set of (multi-version or original) forward rules $R$ is maximal if no rule from $R$ is applicable to $G'$.
\end{definition}

\subsection{Triple Graph Grammars}

Triple graph grammars were initially presented by Schuerr \cite{Sch94_2_ref}. This report is based on the slightly adapted version introduced in \cite{Giese+2014}.

In \cite{Giese+2014}, a triple graph grammar (TGG) relates a source and a target modeling language via a correspondence modeling language and is characterized by a set of TGG rules. A TGG rule is defined by a graph production that simultaneously transforms connected graphs from the source, correspondence and target modeling language into a consistently modified graph triplet. The set of TGG rules has to include a dedicated axiom rule, which has a triplet of empty graphs as its left-hand side and practically defines a triplet of starting graphs via its right-hand side.

The left-hand side of a TGG rule $r = L \xrightarrow{r} R$ can be divided into the source, correspondence, and target domains $L_S$, $L_C$, and $L_T$ respectively, with $L_S \subseteq L$, $L_C \subseteq L$, and $L_R \subseteq L$ and $L_S \uplus L_C \uplus L_R = L$. The right-hand side can similarly be divided into three domains $R_S$, $R_C$, and $R_T$. The type graph for graph triplets and TGG rules is hence given by the union of the type graphs defining the source, correspondence, and target language along with additional edges connecting nodes in the correspondence language to nodes in the source and target language.

Figure \ref{fig:example_tgg_rule} shows a TGG rule for linking the language for abstract syntax graphs given by the type graph in Figure \ref{fig:example_graph_type_graph} to a modeling language for class diagrams given by the type graphs $TT$ in Figure \ref{fig:example_type_graphs_tgg}, using the correspondence language $TC$ from Figure \ref{fig:example_type_graphs_tgg}. The rule simultaneously creates a $FieldDecl$ and $TypeAccess$ along with associated edges in the source domain (labeled S) and a corresponding $Association$ with associated edges in the target domain (labeled T), which are linked via a newly created correspondence node of type $CorrField$ in the correspondence domain (labeled C).

\begin{figure}
\centering
\includegraphics[width=\textwidth]{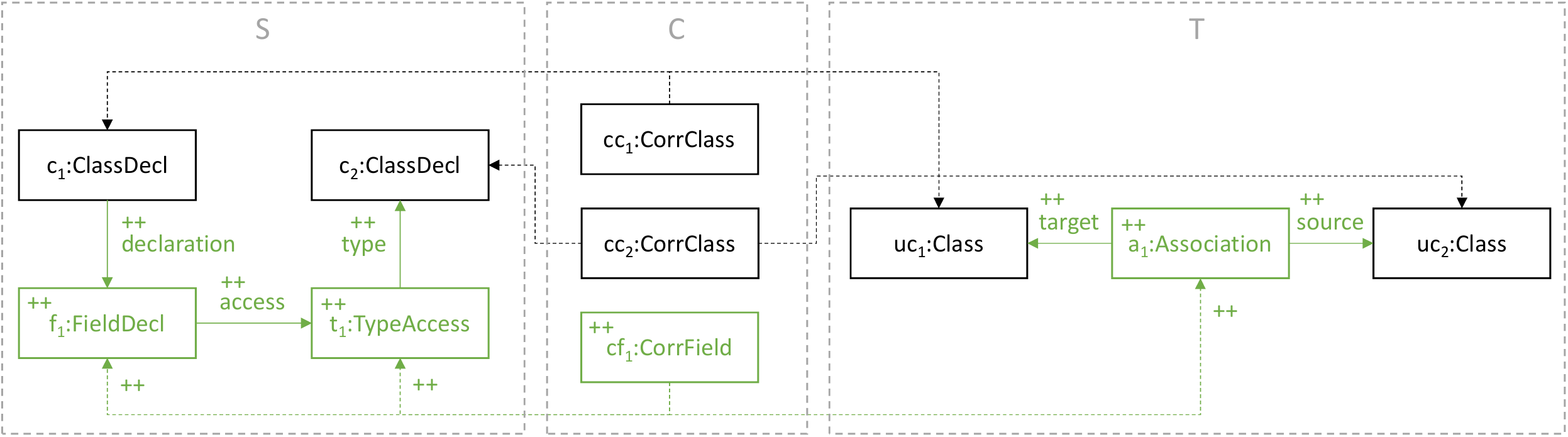}
\caption{example TGG rule in shorthand notation} \label{fig:example_tgg_rule}
\end{figure}

\begin{figure}
\centering
\includegraphics[width=0.4\textwidth]{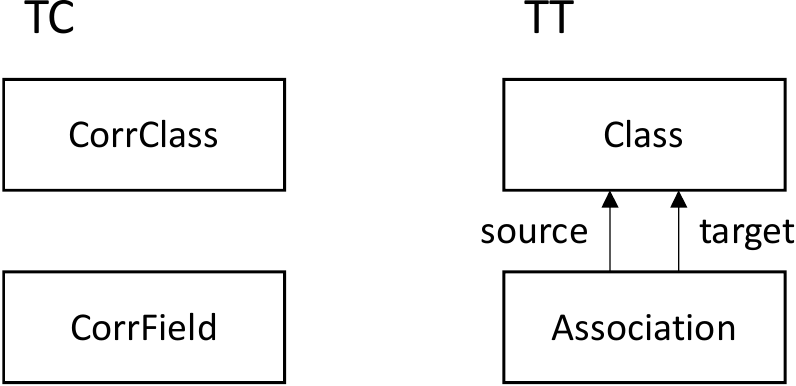}
\caption{example type graphs for the TGG rule in Figure \ref{fig:example_tgg_rule}} \label{fig:example_type_graphs_tgg}
\end{figure}

TGGs can be employed to transform a model of the source language into a model of the target language. This requires the derivation of so-called forward rules from the set of TGG rules. A forward rule for a TGG rule $r = L \xrightarrow{r} R$ can be constructed as $r^F = L^F \xleftarrow{id} L^F \xrightarrow{r^F} R$, where $L^F = L \cup (R_S \setminus r(L))$ and $r^F = r \cup id$, with $id$ the identity morphism. Intuitively, $r^F$ already requires the existence of the elements in the source domain that would be created by an application of $r$ and only creates elements in the correspondence and target domain. In the following, we also denote the subgraph of a forward rule that corresponds to the subgraph that is newly transformed by the rule by $L^T = L^F \setminus L$.

Additionally, the derivation of a forward rule requires a technical extension to avoid redundant translation of the same element. Therefore, a dedicated \emph{bookkeeping node}, which is connected to every currently untranslated source element via a \emph{bookkeeping edge}, is introduced. Then, a bookkeeping node and bookkeeping edges to all elements in $L^T$ are added to the forward rule's left-hand side. The bookkeeping node is also added to the rule's glueing graph and right-hand side. Additionally, negative application conditions are added to $L^F$, which ensure that for a match $m$ from $L^F$ into $SCT$, $\forall x \in L^F \setminus L^T: \nexists b \in B^{SCT}: t_B^{SCT} = m(x)$.

The application of the forward rule via $m$ thus requires that elements in $m(L^T)$ are untranslated, as indicated by the existence of bookkeeping edges, and marks these elements as translated by deleting the adjacent bookkeeping edges. Elements in $m(L^F \setminus L^T)$ are in contrast required to already be translated. Note that, in order to allow bookkeeping edges between the bookkeeping node and regular edges, a slightly extended graph model is used, which is detailed in \cite{giese2010toward}.

Figure \ref{fig:example_forward_rule} shows the forward rule derived from the TGG rule in Figure \ref{fig:example_tgg_rule}. The elements $f_1$ and $t_1$ and adjacent edges are no longer created but preserved instead. Also, the rule requires bookkeeping edges to $f_1$, $t_1$, and adjacent edges, and contains NACs that forbid the existence of bookkeeping edges to $c_1$ and $c_2$. However, this bookkeeping mechanism is omitted in the figure for readability reasons. The rule's application then deletes the bookkeeping edges to $f_1$, $t_1$, and their adjacent edges, and creates the corresponding elements in the target domain along with the linking node $cf_1$ in the correspondence domain.

\begin{figure}
\centering
\includegraphics[width=\textwidth]{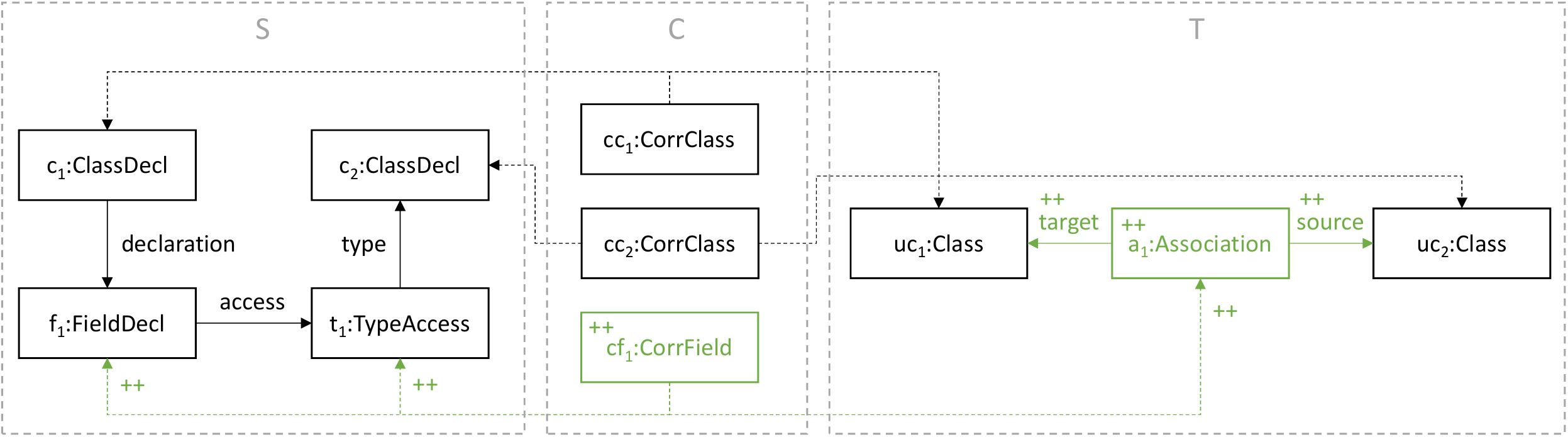}
\caption{example forward rule derived from the TGG rule in Figure \ref{fig:example_tgg_rule}, with the bookkeeping mechanism omitted for readability reasons} \label{fig:example_forward_rule}
\end{figure}

TGGs can also be used to perform a transformation from the target to the source language by means of similarly derived backward rules. In the following, we will focus on the forward case. However, the backward case simply works analogously.

A TGG without any critical pairs \cite{Ehrig+2006} among its rules is called \emph{deterministic} \cite{Giese+2014}. A forward transformation with a deterministic TGG can be executed via an operation $trans^{F}$, which simply applies the TGG's forward rules for as long as there is a match for any of them, with the order of rule applications not affecting the final result due to the absence of critical pairs. Specifically, for a deterministic TGG with a set of forward rules $R$ and a starting model triple $SCT$, any maximal rule transformation sequence $SCT \rightarrow^{R} SCT'$ constitutes a correct model transformation if it deletes all bookkeeping edges in $SCT$. Note that, if $SCT \rightarrow^{R} SCT'$ satisfies this bookkeeping criterion, every other possible maximal rule transformation sequence for $SCT$ and $R$ also satisfies the bookkeeping criterion. In this report, we will focus on such deterministic TGGs, which allow for efficient practical implementations that avoid potentially expensive undoing of forward rule applications and backtracking \cite{Giese+2014}.


\subsection{Multi-version Models}

In this report, we consider models in the form of typed graphs. A model modifications can in this context be represented by a span of morphisms $M \leftarrow K \rightarrow M'$, where $M'$ is the original model, which is modified into a changed model $M'$ via an intermediate model $K$, similar to a graph transformation step \cite{taentzer2014fundamental}. A \emph{version history} of a model is then given by a set of model modifications $\Delta^{M_{\{1,...,n\}}}$ between models $M_1, M_2, ..., M_n$ with type graph $TM$. We call a version history with a unique initial version and acyclic model modification relationships between the individual versions a \emph{correct} version history.

In \cite{Barkowsky2022}, we have introduced \emph{multi-version models} as a means of encoding such a version history in a single consolidated graph. Therefore, an adapted version of $TM$, $TM_{mv}$, is created. To represent model structure, $TM_{mv}$ contains a node for each node and each edge in $TM$. Source and target relationships of edges in $TM$ are represented by edges in $TM_{mv}$. In addition, a $version$ node with a reflexive $suc$ edge is added to $TM_{mv}$, which allows the materialization of the version history's version graph. The version graph and the model structure are linked via $cv_v$ and $dv_v$ edges from each node $v$ in $TM_{mv}$ to the $version$ node.

The result of the adaptation of the type graph from Figure \ref{fig:example_graph_type_graph} is displayed in Figure \ref{fig:example_type_graph_adapted}. Note that $cv$ and $dv$ edges are omitted for readability reasons.

\begin{figure}
\centering
\includegraphics[width=0.75\textwidth]{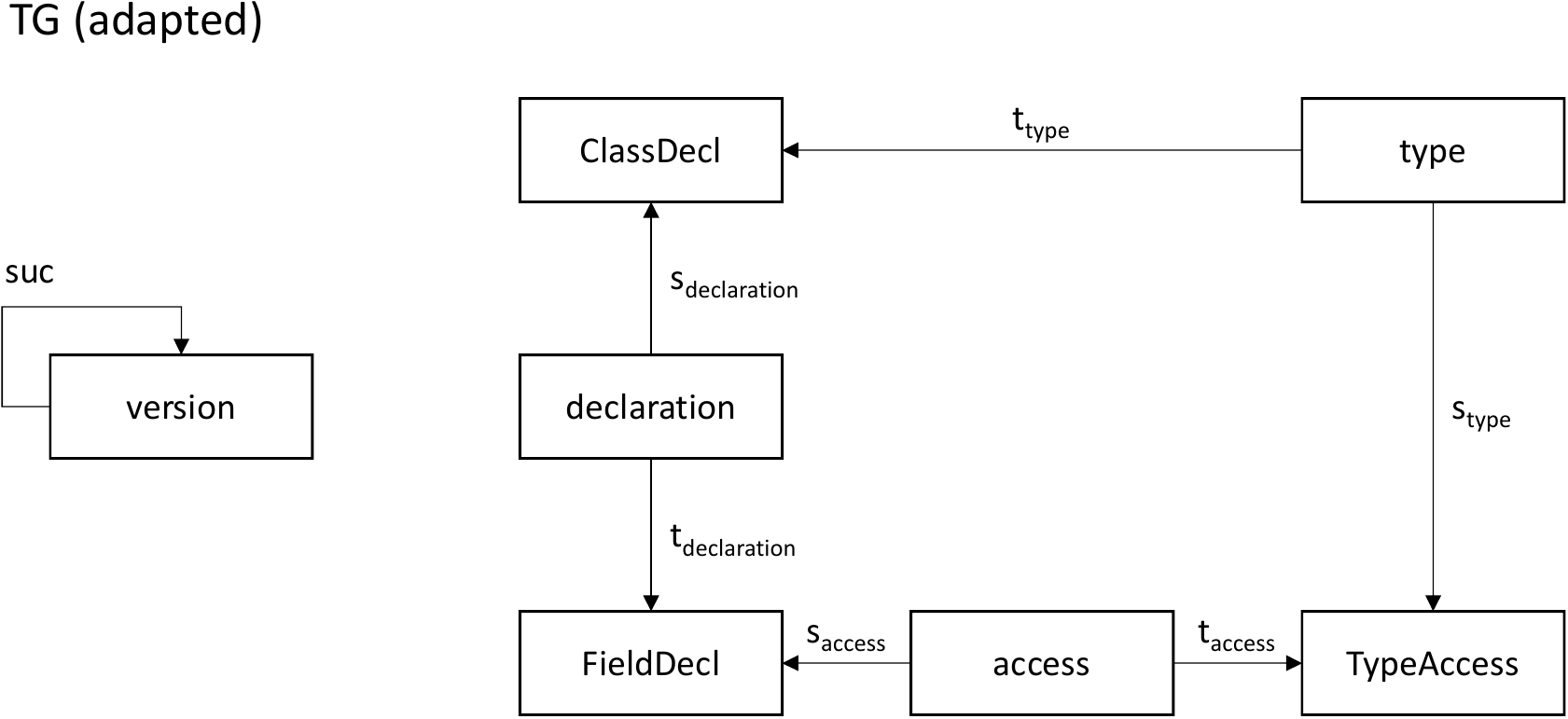}
\caption{example adapted type graph derived from the type graph in Figure \ref{fig:example_graph_type_graph}, with $cv$ and $dv$ edges omitted for readability reasons} \label{fig:example_type_graph_adapted}
\end{figure}

$TM_{mv}$ allows the translation of $\Delta^{M_{\{1,...,n\}}}$ into a single typed graph $MVM$ conforming to $TM_{mv}$, which is called a \emph{multi-version model}, via a procedure $comb$. This yields a bijective function $origin: V^{MVM} \rightarrow \bigcup_{i \in \{1, 2, ..., n\}} V^{M_i} \cup E^{M_i}$ mapping the vertices in $MVM$ to their respective original element. An individual model versions can be extracted from $MVM$ via the projection operation $proj(MVM, i) = M_i$. Finally for a vertex $v_{mv} \in V^{MVM}$, the set of model versions that include the element $origin(v_{mv})$ can be computed via the function $p$, with $p(v_{mv}) = \{M_i \in \{M_1, M_2, ..., M_n\} |  origin(v_{mv}) \in M_i\}$.


\section{Derivation of Multi-version Transformation\\ Rules from Triple Graph Grammars} \label{sec:mv_rules_derivation}

The transformation of the individual model versions encoded in a multi-version model with a triple graph grammar can trivially be realized via the projection operation $proj$. However, the multi-version model may in practice afford a more compact representation compared to an explicit enumeration of all model versions, as derived via $proj$.

In such practical application scenarios, operations concerning all model versions that directly work on the multi-version model may therefore also perform better regarding execution time than the corresponding operations on individual model versions, as we have already demonstrated for the case of pattern matching for checking the well-formedness of all model versions in a version history \cite{Barkowsky2022}. Since pattern matching also constitutes an important task in model transformation via triple graph grammars, a direct, joint translation of all model versions based on the multi-version model representation seems desirable.

Given a triple graph grammar $TGG$, graph transformation rules for the joint translation of all source or target model versions encoded in a multi-version model can be derived from the regular translation rules in a straightforward manner. In the following, we will discuss the deriviation for forward translation. Rules for the backward case can be derived analogously.

First, the adapted multi-version type graph for the TGG's merged source, correspondence and target type graph is created via the translation procedure described in \cite{Barkowsky2022}. The resulting adapted type graph $TG_{mv}$ for multi-version models is extended by two additional edges, $ucv_v$ and $udv_v$, for each node $v$ in the source domain of the merged type graph. Source and target of these edges are given by $s^{TG_{mv}}(ucv_v) = s^{TG_{mv}}(udv_v) = v$ and $t^{TG_{mv}}(ucv_v) = t^{TG_{mv}}(udv_v) = version$, where $version$ is the dedicated version node in the adapted type graph.

Analogously to the bookkeeping edges in the original typegraph, these edges will be used in the translation process to encode in which versions an element represented by a node $v_{mv}$ with type $v$ has already been translated. We therefore define the set of versions where $v_{mv}$ has not been translated yet $u(v_{mv})$ analogously to the set of versions $p(v_{mv})$ where $v_{mv}$ is present, except that $ucv_v$ and $udv_v$ replace $cv_v$ and $dv_v$ in the definition.

Then, for each forward rule $r = L \xleftarrow{l} K \xrightarrow{r} R$ a corresponding multi-version forward rule is created via a procedure $adapt$, with $adapt(r) = trans'(L) \xleftarrow{l_{mv}} trans'(K) \xrightarrow{r_{mv}} trans'(R)$. The vertex morphism of $l_{mv}$ is given by $l_{mv}^V = origin^{-1}\,\circ\,l\,\circ\,origin$ and the edge morphisms by $l_{mv}^E = s \circ origin^{-1} \circ l^E \circ origin \circ s^{-1}$ and $l_{mv}^E = t \circ origin^{-1} \circ l^E \circ origin \circ t^{-1}$ for all edges representing source respectively target relationships. $r_{mv}$ is constructed analogously.

The $trans'$ procedure is a minor adaptation of the $trans$ procedure in \cite{Barkowsky2022}, which ignores the bookkeeping node, bookkeeping edges, and negative application conditions, but otherwise works analogously. The bookkeeping mechanism is instead translated into the additional constraint $P \neq \emptyset$ over $trans'(L)$, where $P = (\bigcap_{v_{mv} \in V^{trans'(L)}} p(v_{mv}) \cap \bigcap_{v_{mv} \in origin^{-1}(L^T)} u(v_{mv})) \setminus \bigcup_{v_{mv} \in V^{trans'(L)}} u(v_{mv})$.

The application of the adapted rule additionally creates outgoing $cv$ and $dv$ edges for all vertices $v^C_{mv} \in V^{trans(R)} \setminus (origin^{-1}\,\circ\,r\,\circ\,origin)(trans(K))$ to realize the assignment $p(v^C_{mv}) \coloneqq P$. Furthermore, for $v_{mv} \in origin^{-1}(r(l^{-1}(L^T)))$, the application also adds and deletes outgoing $ucv$ and $udv$ edges to realize the modification $u(v_{mv}) \coloneqq u(v_{mv}) \setminus P$.

Note that, since the computation of the $p$ and $u$ sets requires considedring paths of arbitrary length, these computations cannot technically be defined as part of the graph transformation but have to be realized externally.

For the set of forward rules $R$, the corresponding set of multi-version forward rules is then simply given by $R_{mv} = \{adapt(r) | r \in R\}$.


\section{Execution of Multi-version Transformations} \label{sec:mv_transformation_execution}

The forward transformation of all model versions encoded in a multi-version model $MVM$ according to a specified TGG can jointly be performed via the TGG's set of multi-version forward rules.

In a first step, all $ucv$ and $udv$ edges present in $MVM$ are removed. Then, for each edge $e_{cv} \in E^{MVM}$ with $type(e_{cv}) = cv_{x}$ and $s^{MVM}(e_{cv})$, an edge $e_{ucv}$ with $type(e_{ucv}) = ucv_{x}$ and $s^{MVM}(e_{cv}) = s^{MVM}(e_{ucv})$ and $t^{MVM}(e_{cv}) = t^{MVM}(e_{ucv})$ is created. For all $dv$ edges, corresponding $udv$ edges are created analogously. Thus, after the creation of the $ucv$ and $udv$ edges, it holds that $\forall v_{vm} \in V^{MVM}: u(v_{vm}) = p(v_{vm})$.

Subsequently, the simultaneous transformation of all model versions encoded in $MVM$ is performed similarly to the regular transformation of a single model version via the TGG. More specifically, the adapted forward rules of the TGG are simply applied to $MVM$ until no such rule is applicable anymore.

In the following, we will argue that this transformation approach is correct in the sense that it yields the same result as the transformation of an individual model version via the regular forward rules.

Therefore, we extend the projection operation $proj$ from \cite{Barkowsky2022} to a bookkeeping-sensitive variant.

\begin{definition} \textit{(Bookkeeping-sensitive Projection)}
For a multi-version model $MVM$ with version graph $V$ and version $t \in V^V$, the bookkeeping-sensitive projection operation works similarly to the regular projection operation $proj$, except that it also adds a bookkeeping node and bookkeeping edges to an element $origin(v)$ iff $t \notin u(v)$ for all $v \in V^{MVM}$. We also denote the result of the bookkeeping-sensitive projection operation by $MVM[t] = proj^{M}(MVM, t)$.
\end{definition}

We also define two sets that represent the bookkeeping during the transformation process.

\begin{definition} \textit{(Bookkeeping Set)}
For a model $M$, we denote the set of translated elements (vertices and edges) by $B(M) = \{x \in M | \nexists b \in E'^{M}: t'^{M} = x\}$, with $E'^{M}$ the set of bookkeeping edges in $M$ and $t'^{M}$ the target function for bookkeeping edges. We also call $B(M)$ the \emph{bookkeeping set} of $M$.
\end{definition}

\begin{definition} \textit{(Projection Bookkeeping Set)}
For a multi-version model $MVM$ and version $t \in V^V$, with $V$ the version graph, we denote the set of already handled elements (vertices and edges) in $MVM[t]$ by $B_{mv}(MVM[t]) = \{x \in MVM[t] | t \notin u(proj^{-1}(x))\}$. We also call $B_{mv}(MVM[t])$ the \emph{projection bookkeeping set} of $MVM[t]$.
\end{definition}

The following theorem states that, at the start of the transformation process via adapted forward rules, the prepared multi-version model via the bookkeeping-sensitive projection correctly encodes the starting situation for the translation of the individual model versions.

\begin{theorem} \label{the:correctness_bookkeeping_projection}
Given a multi-version model $MVM$ encoding a version history with model versions $M_1, M_2, ..., M_n$ such that $\forall v_{vm} \in V^{MVM}: u(v_{vm}) = p(v_{vm})$, it holds that
$\forall t \in \{1, 2, ..., n\}: MVM[t] = init_F(M_t)$
up to isomorphism, where $init_F(SCT_t)$ denotes the graph with bookkeeping resulting from the preparation of $M_t$ for the regular forward transformation process, that is, the graph $M_t$ with an added bookkeeping node and bookkeeping edges to all elements in $M_t$.
\end{theorem}

\begin{proof}
Follows directly from the fact that $\forall t \in \{1, 2, ..., n\}: proj(MVM, t) = M_t$, which has been shown in \cite{Barkowsky2022}, and the definition of the bookkeeping-sensitive projection operation.
\end{proof}

By Theorem \ref{the:correctness_bookkeeping_projection}, we also get the following corollary:

\begin{corollary} \label{cor:correctness_bookkeeping_projection_sets}
Given a multi-version model $MVM$ encoding a version history with model versions $M_1, M_2, ..., M_n$ such that $\forall v_{vm} \in V^{MVM}: u(v_{vm}) = p(v_{vm})$, it holds that
$\forall t \in \{1, 2, ..., n\}: B_{mv}(MVM[t]) = B(init_F(M_t))$
up to isomorphism, where $init_F(SCT_t)$ denotes the graph with bookkeeping resulting from the preparation of $M_t$ for the regular forward transformation process, that is, the graph $M_t$ with an added bookkeeping node and bookkeeping edges to all elements in $M_t$.
\end{corollary}

\begin{proof}
Follows directly from Theorem \ref{the:correctness_bookkeeping_projection} and the definition of bookkeeping set and projection bookkeeping set.
\end{proof}

We now show that a multi-version rule is applicable to a multi-version model iff the corresponding regular rule is applicable to all individual model versions affected by the rule application.

\begin{theorem} \label{the:mv_rule_applicability}
A multi-version forward rule $r_{mv} = L_{mv} \leftarrow K_{mv} \rightarrow R_{mv}$ is applicable to a multi-version model triple $SCT_{mv}$ with bookkeeping via match $m$, if and only if for all $t \in P$, the associated original forward rule $r = L \leftarrow K \rightarrow R$ is applicable to $SCT_{mv}[t]$ via match $trans(m)$, with $P = \bigcap_{v \in V^{L_{mv}}} p(m(v)) \cap \bigcap_{v \in V^{L^{T}_{mv}}} u(m(v))$.
\end{theorem}

\begin{proof}
For a version $t$, as we have already shown in \cite{Barkowsky2022}, the match $m : L_{mv} \rightarrow SCT_{mv}$ has a corresponding match $trans(m) : L \rightarrow SCT_{mv}[t]$ if and only if $t \in \bigcap_{v \in V^{L_{mv}}} p(m(v))$. Furthermore, due to the definition of $P$ and the construction of $r_{mv}$, all elements in $m(trans(m)(L^{T}))$ have an adjacent bookkeeping edge in $SCT_{mv}[t]$ iff $t \in \bigcap_{v \in V^{L^{T}_{mv}}} u(m(v))$. Similarly, all elements in $m(trans(m)(L \setminus L^{T}))$ have no adjacent bookkeeping edge in $SCT_{mv}[t]$ iff $t \notin \bigcup_{v \in V^{L_{mv} \setminus L^{T}_{mv}}} u(m(v))$. Since $r$ and $r_{mv}$ delete no vertices, the dangling condition is trivially satisfied for $r$ and the match $trans(m)$. $r_{mv}$ is hence applicable to $SCT_{mv}$ via $m$, with $t \in P$, iff $r$ is applicable to $SCT_{mv}[t]$ via $trans(m)$.
\end{proof}

We can now show the equivalence of a single multi-version rule application to a multi-version model to the application of the corresponding regular rule to all affected model versions.

\begin{theorem} \label{the:mv_rule_application}
For an application $SCT_{mv} \rightarrow^{r_{mv}}_{m} SCT_{mv}'$ of a multi-version forward rule $r_{mv} = L_{mv} \leftarrow K_{mv} \rightarrow R_{mv}$ with original forward rule $r = L \leftarrow K \rightarrow R$ to a multi-version model triple $SCT_{mv}$ with bookkeeping and version graph $V$ via match $m$, it holds that $\forall t \in P: SCT_{mv}'[t] = SCT' \wedge B_{mv}(SCT_{mv}'[t]) = B(SCT')$ up to isomorphism, with the corresponding application $SCT_{mv}[t] \rightarrow^{r}_{trans(m)} SCT'$. Furthermore, $\forall t \in V^V \setminus P: SCT_{mv}'[t] = SCT_{mv}[t] \wedge B_{mv}(SCT_{mv}'[t]) = B(SCT_{mv}[t])$ up to isomorphism, where $P = \bigcap_{v \in V^{L_{mv}}} p(m(v)) \cap \bigcap_{v \in V^{L^{T}_{mv}}} u(m(v))$.
\end{theorem}

\begin{proof}
Disregarding bookkeeping edges, all forward rules and thus also the adapted forward rules are productions. Due to the construction of the adapted forward rules, all elements created by the rule's application are only mv-present in $SCT_{mv}'$ for the versions in $P$. Therefore, for all remaining versions, $SCT_{mv}[t]$ contains the same elements as $SCT_{mv}'[t]$. An isomorphism $iso: SCT_{mv}[t] \rightarrow SCT_{mv}'[t]$ is hence trivially given by the identity in this case. Since the application of $r_{mv}$ only changes the projection bookkeeping sets for versions in $P$, $B_{mv}(SCT_{mv}'[t]) = B(SCT_{mv}[t])$ with isomorphism $iso$.

It thus holds up to isomorphism that $\forall t \in V^V \setminus P: SCT_{mv}'[t] = SCT_{mv}[t] \wedge B_{mv}(SCT_{mv}'[t]) = B(SCT_{mv}[t])$.

The application of $r_{mv}$ to $SCT_{mv}$ yields a comatch $n : R_{mv} \rightarrow SCT_{mv}'$ and the associated application of $r$ to $SCT_{mv}[t]$ similarly yields a comatch $n' : R \rightarrow SCT'$ for any $t \in P$.

An isomorphism $iso: SCT_{mv}'[t] \rightarrow SCT'$ can then be constructed as follows: Since $r_{mv}$ is a production, $SCT_{mv}$ is a subgraph of $SCT_{mv}'$ and hence $SCT_{mv}[t]$ is also a subgraph of $SCT_{mv}'[t]$. Since $r$ is a production, $SCT_{mv}[t]$ is also a subgraph of $SCT'$. Isomorphic mappings for $SCT_{mv}[t]$ between $SCT_{mv}'[t]$ and $SCT'$ are thus simply given by the identity. This leaves only the elements in $n(R_{mv} \setminus L_{mv})$ and the elements in $n'(R \setminus L)$ unmapped. Due to the construction of $r_{mv}$ being unique up to isomorphism, $n$ and $n'$ being monomorphisms, and $trans$ and $origin$ being bijections, the remaining isomorphic mappings are given by $n' \circ trans \circ n^{-1} \circ origin$. Note that for elements in $n(L_{mv})$, the definition of $iso$ via identity and $n' \circ trans \circ n^{-1} \circ origin$ is redundant but compatible.

Due to the definition of bookkeeping-sensitive projection, bookkeeping set, and projection bookkeeping set, it holds that $B(SCT_{mv}[t]) = B_{mv}(SCT_{mv}[t])$ and thus $B_{mv}(SCT_{mv}[t]) = B(SCT_{mv}[t]))$. Compared to $B_{mv}(SCT_{mv}[t])$, the application of $r_{mv}$ only changes the projection bookkeeping set $B_{mv}(SCT_{mv}'[t])$ by adding the elements in $trans(m(L^{T}_{mv}))$. The modification to $B_{mv}(SCT_{mv}'[t])$ hence corresponds to the modification of the bookkeeping set $B(SCT')$ by the application of $r$ via $trans(m)$ for the isomorphism $iso$ due to the construction of $r_{mv}$.

It thus holds that $\forall t \in P: SCT_{mv}'[t] = SCT' \wedge B_{mv}(SCT_{mv}'[t]) = B(SCT')$.
\end{proof}

Based on Theorem \ref{the:mv_rule_application} for individual rule applications, we get the following corollary for sequences of rule applications:

\begin{corollary} \label{cor:mv_rule_application_sequence}
For a TGG with associated set of forward rules $R$ and multi-version forward rules $R_{mv}$ and a multi-version model triple $SCT_{mv}$ with bookkeeping and version graph $V$, there is a sequence of rule applications $SCT_{mv} \rightarrow^{R_{mv}} SCT_{mv}'$ if and only if for all $t \in V^V$, there is a sequence of rule applications $SCT_{mv}[t] \rightarrow^{R} SCT'$ with $SCT_{mv}'[t] = SCT' \wedge iso(B_{mv}(SCT_{mv}'[t])) = B(SCT')$, where $iso$ is an isomorphism from $SCT_{mv}'[t]$ into $SCT'$.
\end{corollary}

\begin{proof}
We prove the corollary by induction over the length of the multi-version rule application sequence.

For the base case of application sequences of length 0, the identity morphism and empty application sequences trivially satisfy the corollary.

If there is a sequence of rule applications $SCT_{mv} \rightarrow^{R_{mv}} SCT_{mv}'$ if and only if for all $t \in V^V$, there is a sequence of rule applications $SCT_{mv}[t] \rightarrow^{R} SCT'$ with $SCT_{mv}'[t] = SCT' \wedge iso(B_{mv}(SCT_{mv}'[t])) = B(SCT')$, by Theorem \ref{the:mv_rule_application} we have an extended multi-version sequence $SCT_{mv} \rightarrow^{R_{mv}} SCT_{mv}' \rightarrow^{r_{mv}}_{m} SCT_{mv}''$ and all $t \in V^V$ if and only if for all $t \in V^V$, there is a sequence of regular rule applications $SCT_{mv}[t] \rightarrow^{R} SCT''$ with $SCT_{mv}''[t] = SCT'' \wedge iso(B_{mv}(SCT_{mv}''[t])) = B(SCT'')$.

For all $t \in V^V \setminus P$, where $P = \bigcap_{v \in V^{L_{mv}}} p(m(v)) \cap \bigcap_{v \in V^{L^{T}_{mv}}} u(m(v))$, the corresponding regular rule application sequence $SCT_{mv}[t] \rightarrow^{R} SCT'$ and isomorphism $iso : SCT_{mv}'[t] \rightarrow SCT'$ are also valid for $SCT_{mv}''[t]$ and satisfy the condition on bookkeeping sets, since $SCT' = SCT_{mv}'[t] = SCT_{mv}''[t]$ (up to isomorphism).

In accordance with Theorem \ref{the:mv_rule_application}, there is an extended sequence $SCT_{mv} \rightarrow^{R_{mv}} SCT_{mv}' \rightarrow^{r_{mv}}_{m} SCT_{mv}''$ if and only if for all $t \in P$, the regular rule application sequence $SCT_{mv}[t] \rightarrow^{R} SCT_{mv}'[t]$ can be extended by a rule application $SCT_{mv}'[t] \rightarrow^{r}_{trans(m)} SCT_{mv}''[t]$ that satisfies the condition on bookkeeping sets.

Thus, there is a sequence of rule applications $SCT_{mv} \rightarrow^{R_{mv}} SCT_{mv}' \rightarrow^{r_{mv}}_{m} SCT_{mv}''$ if and only if for all $t \in V^V$, there is a sequence of rule applications $SCT_{mv}[t] \rightarrow^{R} SCT''$ with $SCT_{mv}''[t] = SCT'' \wedge iso(B_{mv}(SCT_{mv}''[t])) = B(SCT'')$.

With the proof for the base case and the induction step, we have proven the validity of the corollary.
\end{proof}

Intuitively, the multi-version forward rules perform an interleaved, parallel transformation of all model versions encoded in $SCT_{mv}$. The application of a multi-version rule $L_{mv} \leftarrow K_{mv} \rightarrow R_{mv}$ corresponds to the application of the original rule to all model versions in $P = \bigcap_{v \in V^{L_{mv}}} p(m(v)) \cap \bigcap_{v \in V^{L^{T}_{mv}}} u(m(v))$ and leaves all other model versions unchanged. Thus, a multi-version rule application effectively extends the corresponding original rule application sequences for all versions in $P$ by the associated original rule application, whereas it represents the ``skipping'' of a step in the sequences of all versions not in $P$.

For a deterministic TGG, a correct translation of source graph $S$ is given by any maximal rule application sequence of forward rules that deletes all bookkeeping edges in the source model. Note that because of the determinism criterion, either every maximal rule application sequences or none of them satisfies the bookkeeping criterion. Correctness of the joint translation of all individual versions via multi-version forward rules is hence given by the following corollary:

\begin{corollary} \label{cor:mv_max_rule_application_sequence}
For a TGG with associated set of forward rules $R$ and multi-version forward rules $R_{mv}$ and a multi-version model triple $SCT_{mv}$ with bookkeeping and version graph $V$, there is a maximal sequence of rule applications $SCT_{mv} \rightarrow^{R_{mv}} SCT_{mv}'$ if and only if for all $t \in V^V$, there is a maximal sequence of regular rule applications $SCT_{mv}[t] \rightarrow^{R} SCT'$ such that $SCT_{mv}'[t] = SCT' \wedge B_{mv}(SCT_{mv}'[t]) = B(SCT')$.
\end{corollary}

\begin{proof}
The existence of a sequence of original rule applications for a sequence of multi-version rule applications and all versions $t \in V^V$ and vice-versa is given by Corollary \ref{cor:mv_rule_application_sequence}. From Theorem \ref{the:mv_rule_applicability}, it follows directly that the multi-version sequence is maximal if and only if the regular sequences are maximal for all $t \in V^V$.
\end{proof}

Thus, for a deterministic TGG and by corollaries \ref{cor:correctness_bookkeeping_projection_sets} and \ref{cor:mv_max_rule_application_sequence}, the result of repeated application of adapted transformation rules to a multi-version model prepared for multi-version translation until a fixpoint is reached is equivalent to the results of repeated application of the original rules to the individual model versions prepared for translation, that is, the results of transforming the individual model versions using the TGG.

We thereby have the correctness of the forward transformation using multi-version forward rules $trans^{F}_{mv}$, which applies multi-version forward rules to a multi-version model with bookkeeping until a fixpoint is reached.

\begin{theorem}
For a correct version history $\Delta^{M_{\{1,...,n\}}}$ and a triple graph grammar with set of forward rules $R$, it holds up to isomorphism that

\begin{equation}
\forall t \in \{1, ..., n\} : trans^{F}_{mv}(init_F(comb(\Delta^{M_{\{1,...,n\}}})), adapt(R))[t] = trans^F(M_t, R)
\end{equation}

\end{theorem}

\begin{proof}
Follows directly from Theorem \ref{the:correctness_bookkeeping_projection} and Corollary \ref{cor:mv_max_rule_application_sequence}.
\end{proof}


\section{Evaluation} \label{sec:evaluation}

In order to evaluate our approach empirically with respect to execution time performance, we have realized the presented concepts in our MoTE2 tool \cite{hildebrandt2014} for TGG-based model transformation, which is implemented in the context of the Java-based Eclipse Modeling Framework \cite{emf} and has been shown to be efficient compared to other existing model transformation tools \cite{hildebrandt2014}.

As an application scenario, we consider the transformation of Java abstract syntax trees to class diagrams. We have therefore modeled this transformation as a TGG with MoTE2 and use the original and our adapted implementation to automatically derive the required forward rules respectively multi-version forward rules.

To obtain realistic source models, we have extracted the version history of one small personal Java project (\emph{rete}, around 50 versions) and one larger open source Java project (\emph{henshin} \cite{arendt2010henshin}, around 2000 versions) from their respective GitHub repositories and have constructed the corresponding history of abstract syntax trees using the MoDisco tool \cite{bruneliere2010modisco}. As input for the solution presented in Sections \ref{sec:mv_rules_derivation} and \ref{sec:mv_transformation_execution}, we have consolidated both version histories into multi-version models using a mapping based on hierarchy and naming.

Our implementation and the employed datasets are available under \cite{implementation}.

Based on this, we run the following model transformations for both repositories and measure the overall execution time\footnote{All experiments were performed on a Linux SMP Debian 4.19.67-2 machine with Intel Xeon E5-2630 CPU (2.3\,GHz clock rate) and 386\,GB system memory running OpenJDK version 11.0.6. Reported execution time measurements correspond to the mean execution time of 10 runs of the respective experiment.} for each of them:

\begin{itemize}
\item \textbf{SVM}: individual forward transformation of all model versions (abstract syntax trees) in the repository using the original MoTE2 implementation
\item \textbf{MVM}: joint forward transformation of all model versions in the repository using a multi-version model encoding and our implementation of the technique presented in Sections \ref{sec:mv_rules_derivation} and \ref{sec:mv_transformation_execution}
\end{itemize}

Note that the SVM strategy would require initial projection operations and a final combination of transformation results to work within the framework of multi-version models. However, for fairness of comparison of the transformation, we do not consider these additional operations in our evaluation.

Figure \ref{fig:execution_times} shows the execution times of the transformations using the two strategies. For both repositories, the transformation based on multi-version models requires substantially less time than the transformation of the individual model versions using the original MoTE2 tool, with a more pronounced improvement for the larger repository (factor 4 for the smaller and factor 74 for the larger repository).

\begin{figure}
\centering
\includegraphics[width=0.8\textwidth]{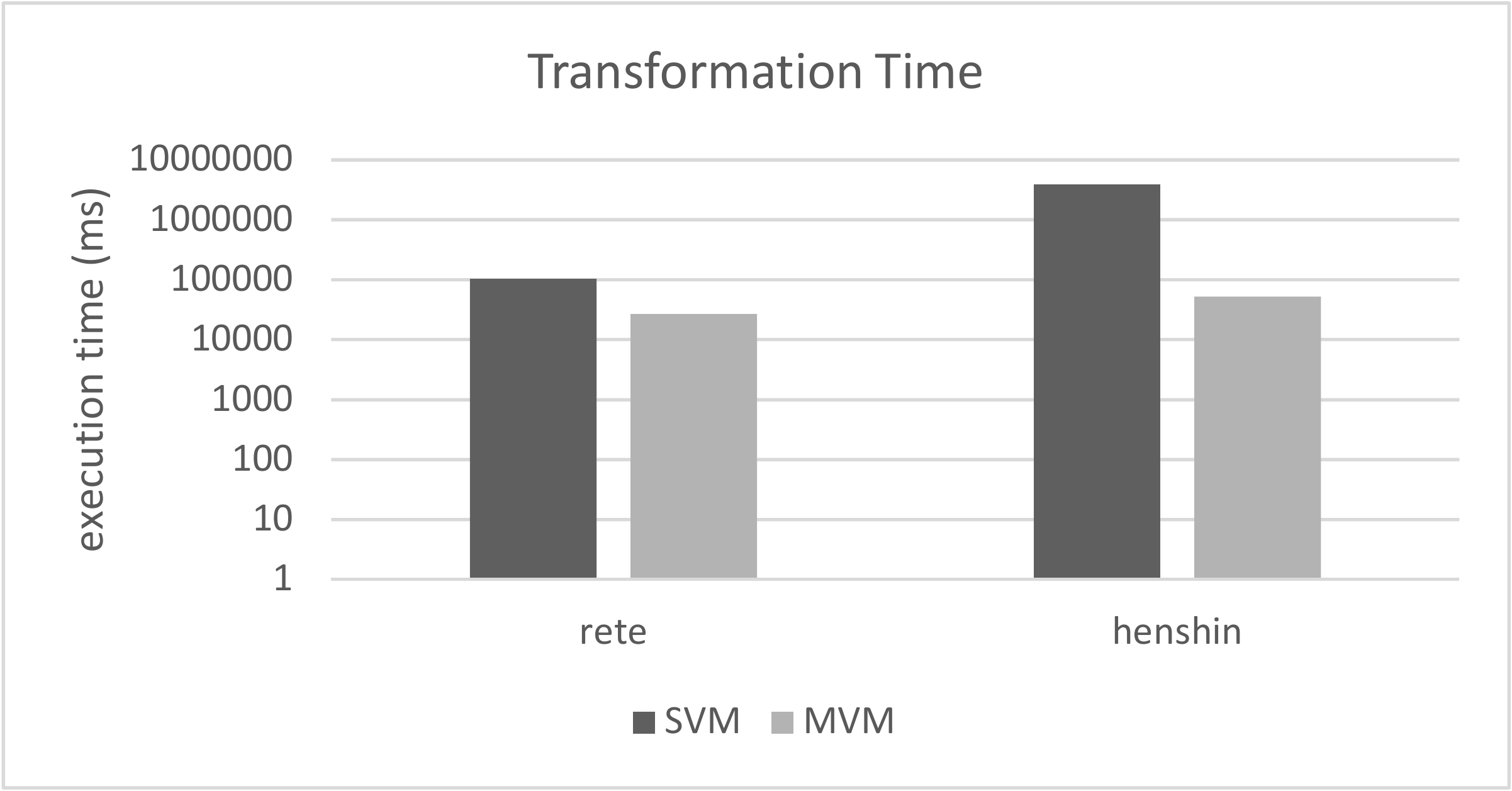}
\caption{execution time measurements for the transformation of all model versions in two different software repositories (logarithmic axis)}\label{fig:execution_times}
\end{figure}

The improvement in efficiency and scalability can likely be explained by two factors: First, SVM has to perform a somewhat expensive initialization step for every indidvidual model version that is to be transformed, whereas MVM only requires one such initialization. Second, many elements in the abstract syntax trees of the repositories are shared between many versions. SVM has to perform a separate transformation, including separate pattern matching, for each model version. In contrast, MVM only performs a transformation including pattern matching over a single multi-version model, the size of which is much smaller than the combined sizes of the encoded model versions, along with efficient search operations over the version graph. Since pattern matching is efficient in this example, that is, pattern matching has a runtime complexity that is linear in the size of the model for the derived forward rules, this results in an improved overall efficiency.

Threats to the internal validity of our experimental results include unexpected behavior of the Java virtual machine such as garbage collection. To address this threat, we have performed multiple runs of all experiments and report the mean execution time of these runs, with the standard deviation always below 5\% of the execution time. To minimize the impact of the concrete implementation on our measurements, we have realized our solution in the framework of the transformation tool we use for comparison and thereby largely use the same execution mechanism.

To mitigate threats to external validity, we use real-world models as the source models of the transformation. However, we remark that our results are not necessarily generalizable to different examples or application domains and make no quantitative claims regarding the performance of our approach.


\section{Related Work} \label{sec:related_work}

The general problem of model versioning has already been studied extensively, both formally \cite{diskin2009model,rutle2009category} and in the form of concrete tool implementations \cite{murta2008towards,koegel2010emfstore}. Several solutions employ a unified representation of a model's version history similar to multi-version models \cite{rutle2009category,murta2008towards}. However, due to the problem definition focusing on the management of different versions of a single model, realising model transformation based on a unified encoding is out of scope for these approaches.

There is also a significant body of previous work on synchronization of concurrently modified pairs of models using triple graph grammars \cite{xiong2013synchronizing,orejas2020incremental}. The focus of these works is the derivation of compatible versions of source and target model that respect the modifications to either of. This report aims to make a step in an orthogonal direction, namely towards allowing living with inconsistencies by enabling developers to temporarily work with multiple modified, possibly conflicting versions of source and target model.

In the context of software product lines, so-called 150\% models are employed to encode different configurations of a software system \cite{10.1007/11561347_28,10.1145/3377024.3377030}. In this context, Greiner and Westfechtel present an approach for propagating so-called variability annotations along trace links created by model transformations \cite{westfechtel2020extending}, explicitly considering the case of transformations implemented via triple graph grammars. A similar approach could also be employed to propagate versioning information and would have the advantage of not requiring any adaptation of the employed rules, type graph, or transformation process. However, not integrating this propagation with the transformation process and only propagating versioning information after the transformation has been executed would mean that certain cases that are covered by our approach could not be handled. The occurence of such cases may hence prevent a possible correct transformation. For instance, under standard TGG semantics, such cases include a model element being translated differently in different model versions based on its context.

In previous work in our group, the joint execution of queries over multiple versions of an evolving model has been considered for both the case with \cite{Barkowsky2022} and without \cite{giese2019metric,sakizloglou2021incremental} parallel, branching development. This report builds on these results, but instead of focusing on pure queries without side-effects considers the case of writing operations in the form of model transformations.


\section{Conclusion}\label{sec:conclusion}

In this report, we have presented a first step in the direction of model transformation on multi-version models in the form of an adaptation of the well-known triple graph grammar formalism that enables the joint transformation of all versions encoded in a multi-version model. The presented approach is correct with respect to the translation semantics of deterministic triple graph grammars for individual model versions, that is, it produces equivalent results. Initial experiments for evaluating the efficiency of our approach demonstrate that our technique can improve performance compared to a na\"ive realization, which simply translates all model versions individually according to a triple graph grammar specification, in a realistic application scenario.

In future work, we plan to build on the presented approach to realize model synchronization for multi-version models, that is, incremental propagation of changes to one or more versions of a source model to the corresponding target model versions. Furthermore, we want to explore the possibility of improving the efficiency of multi-version model transformations via incremental pattern matching for multi-version models. Another interesting direction for future work is the integration of advanced application conditions for the specification of triple graph grammar rules such as nested graph conditions into our approach. Finally, a more extensive evaluation can be conducted to study the scalability of the presented technique in more detail.


\subsection*{Acknowledgements}

This work was developed mainly in the course of the project modular and incremental Global Model Management (project number 336677879), which is funded by the Deutsche Forschungsgemeinschaft.

\printbibliography

@InProceedings{Barkowsky2022,
      AUTHOR = {Barkowsky, Matthias and Giese, Holger},
      TITLE = {{Towards Development with Multi-version Models: Detecting Merge Conflicts and Checking Well-Formedness}},
      YEAR = {2022},
      BOOKTITLE = {Graph Transformation},
      PAGES = {118--136},
      EDITOR = {Behr, Nicolas and Strüber, Daniel},
      ADDRESS = {Cham},
      PUBLISHER = {Springer International Publishing}
}

@article{seibel2010dynamic,
  title={Dynamic hierarchical mega models: comprehensive traceability and its efficient maintenance},
  author={Seibel, Andreas and Neumann, Stefan and Giese, Holger},
  journal={Software \& Systems Modeling},
  volume={9},
  number={4},
  pages={493--528},
  year={2010},
  publisher={Springer}
}

@Article{Finkelstein+1994, 
      AUTHOR = {Finkelstein, Anthony C. W. and Gabbay, Dov and Hunter, Anthony and Kramer, Jeff and Nuseibeh, Bashar}, 
      TITLE = {{Inconsistency Handling in Multiperspective Specifications}}, 
      YEAR = {1994}, 
      JOURNAL = {IEEE Transactions on Software Engineering}, 
      VOLUME = {20}, 
      NUMBER = {8}, 
      PAGES = {569-578}, 
      ADDRESS = {Piscataway, NJ, USA}, 
      PUBLISHER = {IEEE Press},
      doi={10.1109/32.3106670}
}

@book{Ehrig+2006,
AUTHOR = {Ehrig, Hartmut and Ehrig, Karsten and Prange, Ulrike and Taentzer, Gabriele},
TITLE = {{Fundamentals of algebraic graph transformation}},
YEAR = {2006},
PUBLISHER = {Springer},
doi={10.1007/3-540-31188-2},
series={EATCS}
}

@article{taentzer2014fundamental,
  title={{A fundamental approach to model versioning based on graph modifications: from theory to implementation}},
  author={Taentzer, Gabriele and Ermel, Claudia and Langer, Philip and Wimmer, Manuel},
  journal={Software \& Systems Modeling},
  volume={13},
  number={1},
  pages={239--272},
  year={2014},
  publisher={Springer},
  doi={10.1007/s10270-012-0248-x}
}

@InProceedings{Sch94_2_ref,
      AUTHOR = {Sch\"{u}rr, Andy},
      TITLE = {{Specification of graph translators with triple graph grammars}},
      YEAR = {1994},
      MONTH = {June},
      BOOKTITLE = {Proc. of the $20^th$ International Workshop on Graph-Theoretic Concepts in Computer Science},
      VOLUME = {903},
      PAGES = {151--163},
      EDITOR = {Mayr, E.~W. and Schmidt, G. and Tinhofer, G.},
      SERIES = {Lecture Notes in Computer Science},
      ADDRESS = {Herrsching, Germany},
      PUBLISHER = {Spinger Verlag}
}

@Article{Giese+2014,
      AUTHOR = {Giese, Holger and Hildebrandt, Stephan and Lambers, Leen},
      TITLE = {{Bridging the Gap between Formal Semantics and Implementation of Triple Graph Grammars - Ensuring Conformance of Relational Model Transformation Specifications and Implementations}},
      YEAR = {2014},
      JOURNAL = {Software and Systems Modeling},
      VOLUME = {13},
      NUMBER = {1},
      PAGES = {273-299},
      PUBLISHER = {Springer Berlin Heidelberg},
      URL = {http://dx.doi.org/10.1007/s10270-012-0247-y}
}

@article{giese2010toward,
  title={Toward bridging the gap between formal semantics and implementation of triple graph grammars},
  author={Giese, Holger and Hildebrandt, Stephan and Lambers, Leen},
  year={2010},
  volume={Technische Berichte des Hasso-Plattner-Instituts für Digital Engineering an der Universität Potsdam (37)},
  publisher={Universit{\"a}tsverlag Potsdam}
}

@inproceedings{bruneliere2010modisco,
  title={{MoDisco: a generic and extensible framework for model driven reverse engineering}},
  author={Bruneliere, Hugo and Cabot, Jordi and Jouault, Fr{\'e}d{\'e}ric and Madiot, Fr{\'e}d{\'e}ric},
  booktitle={Proceedings of the IEEE/ACM international conference on Automated software engineering},
  year={2010},
  doi={10.1145/1858996.1859032}
}

@PhdThesis{hildebrandt2014,
      AUTHOR = {Hildebrandt, Stephan},
      TITLE = {{On the Performance and Conformance of Triple Graph Grammar Implementations}},
      YEAR = {2014},
      MONTH = {June},
      SCHOOL = {Hasso Plattner Institute at the University of Potsdam }
}

@inproceedings{arendt2010henshin,
  title={Henshin: advanced concepts and tools for in-place EMF model transformations},
  author={Arendt, Thorsten and Biermann, Enrico and Jurack, Stefan and Krause, Christian and Taentzer, Gabriele},
  booktitle={International Conference on Model Driven Engineering Languages and Systems},
  pages={121--135},
  year={2010},
  organization={Springer}
}

@misc{implementation,
  title = {{TGGs for Multi-Version Models Evaluation Artifacts}},
  howpublished = {\url{https://github.com/hpi-sam/TGGs-for-Multi-Version-Models}},
  note = {Last accessed 14 October 2022}
}

@inproceedings{diskin2009model,
  title={{Model-versioning-in-the-large: Algebraic foundations and the tile notation}},
  author={Diskin, Zinovy and Czarnecki, Krzysztof and Antkiewicz, Michal},
  booktitle={2009 ICSE Workshop on Comparison and Versioning of Software Models},
  pages={7--12},
  year={2009},
  organization={IEEE},
  doi={10.1109/CVSM.2009.5071715}
}

@inproceedings{rutle2009category,
  title={{A category-theoretical approach to the formalisation of version control in MDE}},
  author={Rutle, Adrian and Rossini, Alessandro and Lamo, Yngve and Wolter, Uwe},
  booktitle={International Conference on Fundamental Approaches to Software Engineering},
  pages={64--78},
  year={2009},
  organization={Springer},
  doi={10.1007/978-3-642-00593-0_5},
  series={LNTCS},
  volume={5503}
}

@inproceedings{murta2008towards,
  title={{Towards Odyssey-VCS 2: Improvements over a UML-based version control system}},
  author={Murta, Leonardo and Corr{\^e}a, Chessman and Prud{\^e}ncio, Joao Gustavo and Werner, Cl{\'a}udia},
  booktitle={Proceedings of the 2008 international workshop on Comparison and versioning of software models},
  year={2008},
  doi={10.1145/1370152.1370159}
}

@inproceedings{koegel2010emfstore,
  title={{EMFStore: a model repository for EMF models}},
  author={Koegel, Maximilian and Helming, Jonas},
  booktitle={Proceedings of the 32nd ACM/IEEE International Conference on Software Engineering-Volume 2},
  pages={307--308},
  year={2010},
  doi={10.1145/1810295.1810364}
}

@article{xiong2013synchronizing,
  title={Synchronizing concurrent model updates based on bidirectional transformation},
  author={Xiong, Yingfei and Song, Hui and Hu, Zhenjiang and Takeichi, Masato},
  journal={Software \& Systems Modeling},
  volume={12},
  number={1},
  pages={89--104},
  year={2013},
  publisher={Springer}
}

@inproceedings{orejas2020incremental,
  title={Incremental Concurrent Model Synchronization using Triple Graph Grammars.},
  author={Orejas, Fernando and Pino, Elvira and Navarro, Marisa},
  booktitle={FASE},
  pages={273--293},
  year={2020}
}

@article{westfechtel2020extending,
  title={{Extending single-to multi-variant model transformations by trace-based propagation of variability annotations}},
  author={Westfechtel, Bernhard and Greiner, Sandra},
  journal={Software and Systems Modeling},
  volume={19},
  number={4},
  pages={853--888},
  year={2020},
  publisher={Springer},
  doi={10.1007/s10270-020-00791-9}
}

@inproceedings{10.1145/3377024.3377030,
author = {Reuling, Dennis and Pietsch, Christopher and Kelter, Udo and Kehrer, Timo},
title = {{Towards Projectional Editing for Model-Based SPLs}},
year = {2020},
publisher = {Association for Computing Machinery},
address = {New York, NY, USA},
doi = {10.1145/3377024.3377030},
booktitle = {{Proceedings of the 14th International Working Conference on Variability Modelling of Software-Intensive Systems}},
articleno = {25},
numpages = {10},
location = {Magdeburg, Germany},
series = {VAMOS '20}
}

@InProceedings{10.1007/11561347_28,
author="Czarnecki, Krzysztof
and Antkiewicz, Micha{\l}",
editor="Gl{\"u}ck, Robert
and Lowry, Michael",
title={Mapping Features to Models: A Template Approach Based on Superimposed Variants},
booktitle="Generative Programming and Component Engineering",
year="2005",
publisher="Springer Berlin Heidelberg",
address="Berlin, Heidelberg",
pages="422--437",
isbn="978-3-540-31977-1",
doi={10.1007/11561347_28},
series={LNPSE},
volume={3676}
}

@inproceedings{giese2019metric,
  title={{Metric temporal graph logic over typed attributed graphs}},
  author={Giese, Holger and Maximova, Maria and Sakizloglou, Lucas and Schneider, Sven},
  booktitle={International Conference on Fundamental Approaches to Software Engineering},
  pages={282--298},
  year={2019},
  organization={Springer, Cham},
  doi={10.1007/978-3-030-16722-6_16},
  series={LNTCS},
  volume={11424}
}

@article{sakizloglou2021incremental,
  title={{Incremental execution of temporal graph queries over runtime models with history and its applications}},
  author={Sakizloglou, Lucas and Ghahremani, Sona and Barkowsky, Matthias and Giese, Holger},
  journal={Software and Systems Modeling},
  pages={1--41},
  year={2021},
  publisher={Springer},
  doi={10.1007/s10270-021-00950-6}
}

@misc{emf,
  title = {{EMF}},
  howpublished = {\url{https://www.eclipse.org/modeling/emf/}},
  note = {Last accessed 14 October 2022},
  key = {EMF}
}
\end{document}